\documentclass[11pt]{article}
\usepackage{mathrsfs}
\usepackage[scr=boondox,  cal=esstix]   {mathalpha}
\usepackage{amsmath, amssymb, amscd, amsthm, amsfonts}
\usepackage{graphicx}
\usepackage[usenames,dvipsnames]{xcolor}
\usepackage{thmtools,hyperref}
\usepackage{array}
\usepackage{stmaryrd}
\SetSymbolFont{stmry}{bold}{U}{stmry}{m}{n}
\usepackage{bbold}
\usepackage{bm}

\usepackage{tikz}
\usepackage{pgfplots}
\usepackage{cleveref}

\usepackage{xparse}
\hypersetup{
    colorlinks,
    linkcolor={purple},
    citecolor={OliveGreen},
    urlcolor={blue!80!black}
} 
\usepackage[
backend=biber,
style=alphabetic,
maxnames=50,
url=false
]{biblatex}
\usepackage{fixltx2e}

\usepackage{blkarray}
\usepackage{algorithm}
\usepackage{algpseudocode}
\usepackage{multirow}

\usepackage{geometry}
\usepackage[T1]{fontenc}

\pgfplotsset{compat=1.16}

\renewcommand{\epsilon}{\varepsilon}

\title{Approximating $q \rightarrow p$ Norms of Non-Negative Matrices in Nearly-Linear Time}
\date{}
\author{
    \'Etienne Objois\footnote{IRIF, Université Paris Cité, \texttt{objois@irif.fr}} 
    \and
    Adrian Vladu\footnote{CNRS, IRIF, Université Paris Cité, \texttt{vladu@irif.fr}}
}

\newtheorem{theorem}{Theorem}
\newtheorem*{theorem*}{Theorem}
\newtheorem{lemma}[theorem]{Lemma}
\newtheorem*{lemma*}{Lemma}

\newtheorem{corollary}[theorem]{Corollary}
\newtheorem{definition}{Definition}

\newtheorem{claim}{Claim}

\newcommand{\rr}{\mathbb{R}}
\newcommand{\OO}{O}
\newcommand{\OOt}{\widetilde{\OO}}

\newcommand{\abs}[1]{\left\lvert #1 \right\rvert}
\newcommand{\norm}[1]{\left\lVert #1 \right\rVert}

\NewDocumentCommand{\qpnorm}{ O{q} O{p} m }{\left\lVert #3 \right\rVert_{#1 \rightarrow #2}}
\newcommand{\psnorm}[2][p]{\left\lVert #2 \right\rVert_{#1}}
\newcommand{\dnorm}[1]{\left\lVert #1 \right\rVert_{*}}

\global\long\def\AA{\bm{\mathit{A}}}\global\long\def\DD{\bm{\mathit{D}}}\global\long\def\FF{\bm{\mathit{F}}}\global\long\def\GG{\bm{\mathit{G}}}\global\long\def\LL{\bm{\mathit{L}}}\global\long\def\QQ{\bm{\mathit{Q}}}\global\long\def\RR{\bm{\mathit{R}}}

\global\long\def\vf{\bm{f}}

\global\long\def\vu{\bm{u}}
\global\long\def\vv{\bm{v}}
\global\long\def\vx{\bm{x}}
\global\long\def\vy{\bm{y}}
\global\long\def\vz{\bm{z}}
\global\long\def\vq{\bm{\Phi}}
\global\long\def\vdelta{\bm{\delta}}
\global\long\def\ones{\bm{1}}
\global\long\def\zeros{\bm{0}}

\global\long\def\lqp{\ell_{q \rightarrow p}}
\global\long\def\lp{\ell_{p}}
\global\long\def\lq{\ell_{q}}
\global\long\def\Rlo{\mathcal{R}_{\textnormal{lin}}}

\DeclareMathOperator{\np}{NP}
\DeclareMathOperator{\dtime}{DTIME}
\DeclareMathOperator{\bptime}{BPTIME}

\newcommand\sbullet[1][.5]{\mathbin{\vcenter{\hbox{\scalebox{#1}{$\bullet$}}}}}

\DeclareMathOperator{\conv}{conv}
\DeclareMathOperator{\IN}{IN}
\DeclareMathOperator{\OUT}{OUT}
\DeclareMathOperator{\poly}{poly}

\DeclareMathOperator{\opt}{\normalfont \textsc{OPT}}
\DeclareMathOperator{\cost}{\normalfont \textsc{cost}}
\DeclareMathOperator{\CR}{\normalfont \textsc{Competitive-Ratio}}
\DeclareMathOperator{\agg}{\normalfont \textsc{agg}}
\DeclareMathOperator{\nnz}{\normalfont \textnormal{nnz}}

\newcommand{\OBL}{\AA}

\begin{document}

\maketitle

\begin{abstract}
    We provide the first nearly-linear time algorithm for approximating $\ell_{q \rightarrow p}$-norms of non-negative matrices, for $q \geq p \geq 1$. Our algorithm returns a $(1-\varepsilon)$-approximation to the matrix norm in time $\widetilde{O}\left(\frac{1}{q \varepsilon} \cdot \normalfont \textnormal{nnz}(\bm{\mathit{A}})\right)$, where $\bm{\mathit{A}}$ is the input matrix, and improves upon the previous state of the art, which either proved convergence only in the limit [Boyd '74], or had very high polynomial running times [Bhaskara-Vijayraghavan, SODA '11]. Our algorithm is extremely simple, and is largely inspired from the coordinate-scaling approach used for positive linear program solvers.

    We note that our algorithm can readily be used in the [Englert-R\"{a}cke, FOCS '09] to improve the running time of constructing $O(\log n)$-competitive $\ell_p$-oblivious routings. We thus complement this result with a simple cutting-plane based scheme for computing \textit{optimal} oblivious routings in graphs with respect to any monotone norm. Combined with state of the art cutting-plane solvers, this scheme runs in time $\widetilde{O}(n^6 m^3)$, which is significantly faster than the one based on  Englert-R\"{a}cke, and generalizes the $\ell_\infty$ routing algorithm of [Azar-Cohen-Fiat-Kaplan-Räcke, STOC '03].

\end{abstract}
\newpage

\section{Introduction}

We are interested in computing the norm of matrices. We define the $\lqp$-norm of a matrix $\AA \in \rr^{m \times n}$ as
\begin{equation}
    \label{eq:lqp_objective}
    \qpnorm{\AA} = \max_{\vx \neq \zeros} \frac{\psnorm[p]{\AA \vx}}{\psnorm[q]{\vx}}
\end{equation}
where $\psnorm[p]{\vx} = \left(\sum_{i=1}^{n} \abs{\vx_i}^p \right)^{1/p}$, and $q,p \geq 1$. When $q = p = 2$, the quantity $\qpnorm[2][2]{\AA}$ corresponds to the spectral norm. In all generality, $\qpnorm{\AA}$ corresponds to the maximum stretch of the operator $\AA$ from the normed space $\ell_q^n$ to $\ell_p^m$.

The $\lqp$-norm of a matrix appears is different optimization problems. For instance, the case $q = 1$ and $p = \infty$ is the Grothendieck problem, where the goal is to maximize the quantity $\langle \AA \vx, \vy \rangle$ over $\vx,\vy$ such that $\psnorm[\infty]{\vx}, \psnorm[\infty]{\vy} \leq 1$. Our problem can be seen as a generalized version of the Grothendieck problem.
\begin{equation*}
    \qpnorm{\AA} = \max_{\begin{array}{l}
        \psnorm[q]{\vx} \leq 1 \\
        \psnorm[p^*]{\vy} \leq 1 \\
    \end{array}} \langle \vy, \AA \vx \rangle = \qpnorm[p^*][q^*]{\AA^T}\,,
\end{equation*}
where $q^*,p^*$ denotes the dual norms of $q,p$ i.e. they satisfy $1/q + 1/q^* = 1/p + 1/p^* = 1$. Notice that by changing from primal to dual norms, if $q \geq p \geq 1$, then $p^* \geq q^* \geq 1$. Hence we cannot switch the order of $p$ and $q$.

When $1 \leq q < p$, literature refers to $\qpnorm{\AA}$ as a hypercontractive norm. For example, $\qpnorm[2][4]{\AA}$ is linked with quantum information theory and when $p \geq 4$, approximating $\qpnorm[2][p]{\AA}$ would prove wrong the Small-Set Expansion Hypothesis from \cite{RS10}, a variant of Khot's Unique Game Conjecture \cite{Barak_2012}. 

When $p = q$, $\qpnorm[p][p]{\AA}$ is known as the $p$-norm of $\AA$ and has important applications to computing linear oblivious routings for graphs~\cite{ER09}. Noteworthy, \cite{BV11} showed that by being able to approximate the $p$-norm of a non-negative matrix, one can compute in polynomial time $\OO(\log n)$-competitive linear oblivious routings in undirected graphs when the load function is an unknown monotone norm and the aggregation function is an $\lp$-norm on the load vector.

Finally, the non-hypercontractive norm case refers to parameters $p,q$ satisfying $1 \leq p < q$. This problem is linked with robust optimization \cite{Ste05}. The general version of the problem is known to be NP-hard under some conjectures for general matrices (see \cite{BG19} for details). In this paper, we focus on parameters $p,q$ satisfying $q \geq p \geq 1$, as well as the matrix $\AA$ being non-negative.

\subsection{Our Contributions}

We give the first input-sparsity time algorithm to approximate the $\lqp$-norm of a non-negative matrix when $q \geq p \geq 1$. The result is stated in the following theorem.
\begin{theorem}[Informal version of \autoref{thm:algo_qpnorm_better}]
    \label{thm:intro_qpnorm}
    Given a non-negative matrix $\AA \in \rr^{m \times n}$, a positive real $0 < \epsilon \leq 1/2q$, two reals $q \geq p \geq 1$, and a guess $V$ on $\qpnorm{\AA}$, \autoref{alg:qpnorm2} recovers $\vx$ such that 
    \begin{equation*}
        \frac{\psnorm[p]{\AA \vx}}{\psnorm[q]{\vx}} \geq (1-\epsilon)V
    \end{equation*}
    or certifies infeasibility in time $\OOt\left(\frac{\nnz(\AA)}{q \epsilon}\right)$\footnote{$\OOt$ hides poly logarithmic factors in $\epsilon^{-1},n,m,q,p$.}.
\end{theorem}

While the statement of \autoref{thm:intro_qpnorm} concerns approximation of a decision problem, standard techniques can transform an approximation decision algorithm into an approximation algorithm. We refer the reader to \autoref{sec:approximation} for the details which we provide for completeness.

We can use \autoref{thm:intro_qpnorm} for the case $q = p$ to improve the running time of \cite{BV11} by a factor of $\OO(m^3)$. Moreover, if both the load function and the aggregation function are known monotone norms, we provide the first algorithm to compute an optimal oblivious linear routing. Notice that this problem is slightly different from the one solved by \cite{BV11} as we do need knowledge of the load function. Hence the following theorem.

\begin{theorem}[Informal version of \autoref{thm:main_result_ccg}]
    Given a directed or undirected graph $G=(V,E)$, there exists an algorithm to compute an optimal oblivious routing when the cost function is a known monotone norm in time $\OOt(n^6 m^3)$.
\end{theorem}

\subsection{Our Techniques}

We seek to maximize $g : \vx \mapsto \psnorm[p]{\AA \vx}^q/\psnorm[q]{\vx}^q$. Our approach is partly inspired from the long line of work on solvers for positive linear programs~\cite{luby1993parallel,Young01,allen2019nearly,mahoney_et_al:LIPIcs:2016:6333}, as well as certain non-standard instantiations of the same framework in the context of regression problems~\cite{ene2019improved}. Those algorithms are \textit{width-independent}, meaning their running time is at most linear in $\nnz(\AA) + n +m$ (at a cost of a larger dependency on the precision $\epsilon$). By design, those algorithms are different from classical multiplicative weight update algorithms as their running time does not depend on the \textit{width} of the problem.

At each iteration, we scale by $(1+\alpha)$ some coordinates of our iterate to ensure large progress is made. We say that a coordinate is hit at time $t$ if we scale it between iterations $t$ and $t+1$. We do not directly measure the progress on $g$, instead, we separate the progress obtained in the denominator from the progress obtained in the numerator. At each iteration, assuming $\AA$ is non-negative, we want to ensure the ratio between those progresses is larger than $(1-\OO(\epsilon))^q \qpnorm{\AA}^q$ (\autoref{lem:basic_invariant} with $V = \qpnorm[q][p]{\AA}$). This ensures after $T$ iterations that $g(\vx^{(T)})$ is a $(1-\OO(1/\psnorm[q]{\vx^{(T)}}^q))$ approximation of $\qpnorm{\AA}^q$. To do so, we consider a potential function $\vq : \rr^n \rightarrow \rr^n$ (\autoref{def:potentials}) and hit the coordinates whose potential is above a threshold. $\AA$ being non-negative is crucial as it keeps the property that the norm $\vx \mapsto \psnorm[p]{\AA \vx}$ is monotone.

The algorithm converges once $\psnorm[q]{\vx^{(T)}}$ is large. If we were to hit only one coordinate of $\vx^{(t)}$ at iteration $t$ (say the coordinate with largest potential), the running time would depend on the ambient dimension $n$ which would not be a width-independent algorithm. To overcome this issue, we transform our problem into an approximate decision one. Given $V \geq 0$, a guess on $\qpnorm{q}{p}{\AA}$, we either want to find $\vx$ such that $g(\vx) \geq (1-\epsilon)^q V^q$ or certify that $V > \qpnorm[q][p]{\AA}$. We then change the threshold for $\vq$ accordingly. The certification comes from the fact that at each iteration, at least one coordinate of $\vq(\vx)$ should be above $\qpnorm[q][p]{\AA}^q$ (\autoref{lem:infeasibility_full}). Assuming $V \leq \qpnorm[q][p]{\AA}$, there should always be a coordinate with potential larger than $V^q$. 
This modification still does not give a width-independent algorithm, however, $\vq$ behaves ``nicely'' under multiplicative scaling which allows us to upper bound coordinate-wise the value of $\vq(\vx^{(t+1)})/\vq(\vx^{(t)})$ (\autoref{lem:equivalent_potentials_full}).

We use this upper bound to prevent new coordinates from having potential larger than $V^q$. This allows us to guarantee that at iteration $t$, a coordinate with potential larger than $V^q$ also had potential larger than $V^q$ at iterations $0,1, \dots, t-1$. Since its potential was always larger than $V^q$, it was always hit and so it is large. In order to prevent coordinates from having potential larger than $V^q$, we need to also hit coordinates with potential slightly smaller than $V^q$ which is not an issue since we aim for an approximation. This gives a first width-independent algorithm to approximate the $\lqp$-norm of a non-negative matrix in time $\OOt(\nnz(\AA)/q \epsilon)$ or in parallel time $\OOt(1/q \epsilon)$ (\autoref{cor:parallel_running_time}).

The algorithm can also use preconditioning techniques for $\ell_p$ subspace embedding. We first compute a matrix $\AA'$ such that $\psnorm[p]{\AA \vx} = (1 \pm \epsilon) \psnorm[p]{\AA' \vx}$ and $\AA'$ has fewer nonzero entries than $\AA$. Lewis Weight Sampling introduced by Cohen and Peng \cite{CohenP14} and improved for $p > 2$ in \cite{woodruff2022onlinelewisweightsampling} allows to reduce the number of rows of $\AA \in \rr^{m \times n}$ to $\OOt(n/\epsilon^2)$ when $p \leq 2$. This preconditioning technique, paired with the fact that $\qpnorm[q][p]{\AA} = \qpnorm[p^*][q^*]{\AA^T}$ where $1/q + 1/q^* = 1/p + 1/p^* =1$, allows better running time when $m \gg n$ or $n \gg m$. This result is showed in \autoref{cor:algo_qpnorm_better}.

The algorithm is linked with the Perron-Frobenius theorem, in the sense that for $q = p = 2$, if $\AA$ is positive then the maximum of $\vx \mapsto \psnorm[2]{\AA \vx}/\psnorm[2]{\vx}$ is the spectral norm of $\AA$. For positive matrices, \autoref{lem:infeasibility_full} is a generalized version of the Collatz–Wielandt formula which states that $\min_{\vx > \zeros} \max_{i} \frac{\langle \AA_{\sbullet,i}, \AA \vx \rangle}{\vx_i} = \qpnorm[2][2]{\AA}^2$. We generalize this formula to $\ell_p$ norms, and obtain \\ $\min_{\vx > \zeros} \max_{i} \frac{\langle \AA_{\sbullet,i}, (\AA \vx)^{p-1} \rangle}{\vx_i^{p-1}} = \qpnorm[p][p]{\AA}^p$. When $p \neq q$, we just add a scaling factor that depends on $\psnorm[q]{\vx}$ to find the equality.

We then show that an application of our algorithm is to compute a $\OO(\log n)$-competitive linear oblivious routing with a better running time than \cite{BV11}. When the load function is a known monotone norm, we show that we can reduce the problem to finding a saddle point of a bi-linear function.

\subsection{Related Works}

Computing the $\lqp$ norm on general matrices is hard. For arbitrary matrices, there are three known easy cases : $\qpnorm[2][2]{\AA}$, $\qpnorm[q][\infty]{\AA}$ and $\qpnorm[1][p]{\AA}$ (see \cite{Ste05}).

On the hardness front, for arbitrary matrices, computing $\qpnorm{\AA}$ when $1 \leq p < q \leq \infty$ is NP-hard \cite{Ste05}. For $q = p \neq 1,2,\infty$, even if $p$ and the matrix are assumed to have rational entries, it is NP-hard to compute $\qpnorm[p][p]{\AA}$~\cite{HO10}.

If we aim to find an approximation, for arbitrary matrices the problem is also hard to solve to arbitrary precision. For $q \geq p > 2$ (and $2 > q \geq p > 1$), assuming  $\np \notin \dtime(n^{\poly \log (n)})$, the problem cannot be approximated to a factor $2^{(\log n)^{1-\epsilon}}$, for any constant $\epsilon > 0$ \cite{BV11}. \cite{Barak_2012} showed that under the Small-Set Expansion
hypothesis, it is NP-hard to approximate $\qpnorm[2][p]{\AA}$ for arbitrary matrices $\AA$ when $p \geq 4$. Finally, when $2 < q < p < \infty$ (and $1 < q <p < 2$), if $\np \notin \bptime(2^{(\log n)^{\OO(1)}})$, it is NP-hard to approximate $\qpnorm{\AA}$ to a factor $2^{(\log n)^{1-\epsilon}}$ for any constant $\epsilon > 0$ \cite{BG19}. 

For non-negative matrices, \cite{BOYD197495} used a power iteration-type algorithm to approximate $p$-norm of non-negative matrices without any bounds on the time of convergence. When $q \geq p \geq 1$, \cite{BV11} provided the analysis and proved that the power iteration-type algorithm form \cite{BOYD197495} converges in time $\OOt(\frac{n(n+m)^2}{\epsilon}\nnz(\AA))$. Under the same settings, \cite{Ste05} demonstrated that the problem is equivalent to maximizing a concave function ($\vx \mapsto \psnorm[p]{\AA \vx^{\frac{1}{q}}}$) over the unit simplex. When $q < p$, no polynomial time algorithm is known even for non-negative matrices.

The non-negative assumption may be necessary to make the problem easy. Indeed, it is possible to relax MAXCUT into approximating the $\ell_{\infty \rightarrow 1}$ norm of a symmetric positive semi-definite matrix (see \cite{Ste05}). Hence, approximating $\qpnorm[\infty][1]{\AA}$ is NP-hard even when $\AA$ is positive semi-definite and the precision sought is fixed. However, we can easily extend the results on non-negative matrices to matrices $\AA$ such that there exist two diagonal matrices $\LL \in \rr^{m \times m}$ and $\RR \in \rr^{n \times n}$ with $\pm 1$ entries on their diagonals such that $\LL \AA \RR$ is non-negative \cite{Ste05}. Indeed, given $\LL,\RR$, we have that for any $\vx$, $\psnorm[p]{\AA \vx} = \psnorm[p]{\LL \AA \vx}$. Hence, we can apply the algorithm on $\LL \AA \RR$ and consider $\RR \vx$ where $\vx$ is a $(1-\epsilon)$-approximation of $\qpnorm{\LL \AA \RR}$.

\section{Preliminaries}
\label{sec:preliminaries}

\subsection{Notations}

We use lowercase bold letters for vectors and uppercase for matrices. For a vector $\vz$ (resp. a matrix $\AA$), $\vz^T$ (resp. $\AA^T$) denotes the transpose of $\vz$ (resp. of $\AA$). For an integer $i$, $\AA_i$ denotes the $i$\textsuperscript{th} row of $\AA$, and $\AA_{\sbullet,i}$ the $i$\textsuperscript{th} column. We say that a vector or a matrix is non-negative if all its entries are non-negative.  $\langle., . \rangle$ is the usual inner product and $\langle.,. \rangle_F$ denotes the Frobenius inner product. For $p \geq 1$ and a vector $\vz$, we denote by $\psnorm[p]{\vz}$ the $\lp$-norm of $\vz$ defined as
\begin{equation*}
    \psnorm[p]{\vz} := \left(\sum_{i} \abs{\vz_i}^p\right)^{\frac{1}{p}}.
\end{equation*}
For $q,p \geq 1$ and a matrix $\AA$, we denote by $\qpnorm{\AA}$ the $\lqp$-norm of $\AA$ defined as
\begin{equation*}
    \qpnorm{\AA} := \max_{\vz \neq \zeros} \frac{\psnorm[p]{\AA \vz}}{\psnorm[q]{\vz}}.
\end{equation*}

$\zeros$ (resp. $\ones$) denotes the vector with all entries equal to $0$ (resp. $1$). For a vector $\vz$ and a real $q$, $\vz^q$ is the vector where the exponent is applied coordinate wise. For a matrix $\AA$, $\nnz(\AA)$ denotes the number of nonzero elements in $\AA$. We denote by $\Delta_n$ the non-negative simplex $\Delta_n := \lbrace \vz \in \rr^n_{\geq 0} : \sum_i \vz_i = 1 \rbrace$.

\subsection{Convex Optimization}

We would like to emphasize that the problem of maximizing $f(\vx) = \psnorm[p]{\AA \vx}$ on the unit $\lq$ ball corresponds to maximizing a convex function over a convex domain which can not be solved using classical convex optimization tools. Nevertheless, when $\AA$ is non-negative, $q \geq p \geq 1$, we can show that by composing $f$ with a simple concave mapping from the unit simplex $\Delta_n$ to the unit $\lq$ sphere we obtain a concave function that we would like to maximize over a convex domain \cite{Ste05}. To intuit, one should remark that the maximum of $f$ is achieved by a vector in the positive orthant of the unit $\lq$ sphere which we call $S_q$. Moreover, $S_q$ is not a convex set. In fact, for two different vectors $\vx, \vy \in S_q$, and $0 < t < 1$, we have $\psnorm[q]{t \vx + (1-t) \vy} < 1$, hence we are ``losing'' norm by taking the convex combination of two points. That means that even if $f$ is convex, we have the potential to ``increase'' the value of $f$ between any two points of $S_q$ by taking a path in $S_q$. This new function has better chance to be concave than $f$ since its value between two points of $S_q$ is larger. One can remark that if $\vu \in \Delta_n$, then $\vu^\frac{1}{q} \in S_q$. The mapping $g : \vu \mapsto \vu^\frac{1}{q}$ describes such a path since for any $t \in [0,1]$ and $\vu, \vv \in \Delta_n$, we have that $t g(\vu) + (1-t) g(\vv) \in S_q$. As expected, $g$ is entry-wise concave and calculus gives that as long as $q \geq p$, $\vu \mapsto f(g(\vu)) = \psnorm[p]{\AA \vu^{\frac{1}{q}}}$ is concave. We give the proof in \autoref{apps:concave_norm}.

\begin{theorem}[Remark 3.4 from \cite{Ste05}]
    \label{thm:concave_norm}
    Given $\AA \in \rr^{m \times n}_{\geq 0}$, and $q \geq p \geq 1$. We have $f : \vx \mapsto \psnorm[p]{\AA \vx^{\frac{1}{q}}}$ is concave on $\Delta_n$. 
\end{theorem}

This shows that a simple re-parametrization of the problem turns it into a concave maximization problem, which can be efficiently solved using a cutting plane method~\cite{JLLS23GLM}, leading to a running time of $\widetilde{O}(n\cdot \nnz(\AA) + n^3)$. However, as cubic runtime may turn out to be prohibitive, we will focus on more efficient algorithms at the expense of a linear dependence in the error tolerance of the provided solution.

\subsection{Power Iteration}

The previous state-of-the-art algorithm was provided in \cite{BV11}. Let
\begin{equation*}
    f(\vx) = \frac{\psnorm[p]{\AA \vx}}{\psnorm[q]{\vx}}.
\end{equation*}
In order to compute $\qpnorm{\AA}$, one needs to maximize $f$ over $\rr^n$. If we assume $\AA$ is non-negative, then \cite{BV11} provide an algorithm with running time $\OOt(\frac{n(n+m)^2}{\epsilon} \nnz(\AA))$. To find the maximum of $f$, we can compute its gradient $\nabla f$. We have 
\begin{equation*}
    \frac{\partial f}{\partial \vx_i} = \frac{\psnorm[q]{\vx} \psnorm[p]{\AA \vx}^{1-p} \langle \AA_{\sbullet,i}, \abs{\AA \vx}^{p-1} \rangle - \psnorm{\AA \vx} \psnorm[q]{\vx}^{1-q} \abs{\vx_i}^{q-1}}{\psnorm[q]{\vx}^2}.
\end{equation*}
At optimality, $\nabla f = \zeros$ which can be written as 
\begin{equation*}
    \abs{\vx}^{q-1} = \frac{\psnorm[q]{\vx}^q}{\psnorm[p]{\AA \vx}^p} \AA^T \abs{\AA \vx}^{p-1}.
\end{equation*}
Since $\AA$ is non-negative, we know there is a maximum $\vx$ with non-negative coordinates. \cite{BV11} then define the operator $S : \rr^n_{\geq 0} \rightarrow \rr^n_{\geq 0}$ such that $S(\vx) = \left(\AA^T (\AA \vx)^{p-1}\right)^{\frac{1}{q-1}}$. It is easy to show that if $\vx$ is such that $S(\vx) \propto \vx$, then $\vx$ is a critical point and $\qpnorm{\AA} = \frac{\psnorm[p]{\AA \vx}}{\psnorm[q]{\vx}}$. Consider the two potentials $m(\vx) := \min_i S(\vx)_i/\vx_i$ and $M(\vx) = \max_i S(\vx)_i/\vx_i$, \cite{BV11} showed the following lemma. 
\begin{lemma}[Lemma 3.3, \cite{BV11}]
    \label{lem:potentials}
    For any non-negative matrix $\AA$ and positive vector $\vx$, we have 
    \begin{equation*}
        m(\vx)^{q-1} \leq \frac{\psnorm[p]{\AA \vx}^p}{\psnorm[q]{\vx}^q} \leq \qpnorm{\AA}^p \psnorm[q]{\vx}^{p-q} \leq M(\vx)^{q-1}.
    \end{equation*}
\end{lemma}
At optimum, $m(\vx^*)^{q-1} = M(\vx^*)^{q-1} = \psnorm[p]{\AA \vx^*}^p/\psnorm[q]{\vx^*}^q$. Hence, the goal is to reduce the ratio $M(\vx)/m(\vx)$. \cite{BV11} showed that iteratively applying $S$ such that $\vx^{(t+1)} = S(\vx^{(t)})$ outputs $\vx^{(T)}$ a $(1-\epsilon)$-approximation when $T = \OOt(\frac{n(n+m)^2}{\epsilon})$. Hence, the total running time is $\OOt(\frac{n(n+m)^2}{\epsilon} \nnz(\AA))$. In this paper, we provide a much simpler algorithm with running time $\OOt(\nnz(\AA)/q\epsilon)$.

\subsection{Lewis Weight Sampling}

Cohen and Peng introduced in \cite{CohenP14} a fast algorithm to compute a matrix $\AA'$ such that $\psnorm[p]{\AA \vx} \approx_{1+\epsilon} \psnorm[p]{\AA' \vx}$ for any vector $\vx$. The algorithm computes $\AA'$, a matrix containing few rescaled rows of $\AA$. The algorithm is of particular interest in our situation when $\AA \in \rr^{m \times n}$ is such that $m \gg n$. For $1 \leq p \leq 2$, $\OOt(\frac{n}{\epsilon^2})$ rows are sufficient. Moreover, the cost of such an algorithm has only time complexity $\OOt(\nnz(\AA) + n^{\omega})$ where $\omega$ is the time complexity of matrix multiplication. Formally, this gives the following theorem. 

\begin{theorem}[Theorem 1.3 and A.2 from \cite{woodruff2022onlinelewisweightsampling}]
    \label{thm:lewis_sampling}
    Given a matrix $\AA \in \rr^{m \times n}$ and $p \geq 1$, one can compute $\AA'$ such that with probability over $1-\delta$, we have
    \begin{equation*}
        \forall \vx \in \rr^n, \left(1-\epsilon\right) \psnorm[p]{\AA' \vx} \leq \psnorm[p]{\AA \vx} \leq \left(1+\epsilon\right) \psnorm[p]{\AA' \vx}.
    \end{equation*}
    in time $\OOt(\nnz(\AA) + n^{\omega} + \frac{n^{\max \lbrace 1, p/2 \rbrace}}{\epsilon^2})$. Moreover, $\AA'$ has at most 
    \begin{equation*}
        \OO\left(\frac{n^{\max \lbrace 1, p/2 \rbrace}}{\epsilon^2} \left((\log n)^2 \log m + \log \frac{1}{\delta}\right)\right)
    \end{equation*}
    rows.
\end{theorem}

Since the algorithm just samples and rescales rows of $\AA$, if $\AA$ is non-negative, then so is $\AA'$. Hence, it makes Lewis Weight Sampling a possible preconditioning of the input.

\section{Efficiently Computing Induced \texorpdfstring{$\lqp$}{L q to p}-Norms of Non-Negative Matrices}
\label{sec:lp_norm_algo}

\subsection{A Simple \texorpdfstring{$\OOt\left(\frac{n}{q \epsilon} \nnz(\AA)\right)$}{n/q eps nnz(AA)} Algorithm}

We want to solve the following approximate decision problem, which we show is sufficient to solve the approximate optimization problem in \autoref{sec:approximation}. 

\begin{definition}
    Given $V \geq 0$, we solve the approximate decision problem if we either provide $\vx$ such that 
    \begin{equation*}
        \frac{\psnorm[p]{\AA \vx}}{\psnorm[q]{\vx}} \geq (1-\epsilon)V,
    \end{equation*} 
    or certify that $\qpnorm{\AA} < V$.
\end{definition}

We seek to maximize a function $f$ of the form $f : \vx \mapsto \frac{f_n(\vx)}{f_d(\vx)}$ where both $f_n$ and $f_d$ are convex functions. Hence, we have 
\begin{equation*}
    \frac{f_n(\vx + \vdelta) - f_n(\vx)}{f_d(\vx + \vdelta) - f_d(\vx)} \geq \frac{\langle \nabla f_n(\vx), \vdelta \rangle}{\langle \nabla f_d(\vx + \vdelta), \vdelta \rangle}.
\end{equation*}
If we take $f_d : \vx \mapsto \psnorm[q]{\vx}$, we have $\nabla f_d(\vx + \vdelta) = \left(\frac{\vx+ \vdelta}{\psnorm[q]{\vx+  \vdelta}}\right)^{q-1}$ which is not optimal to upper bound by $\nabla f(\vx)$. Instead, we consider $g_{d} : \vx \mapsto \psnorm[q]{\vx}^q$ and we want to maximize $g : \vx \mapsto \left(\frac{\psnorm[p]{\AA \vx}}{\psnorm[q]{\vx}}\right)^q$. We now obtain a tighter bound 
\begin{equation*}
    \frac{\psnorm[p]{\AA (\vx + \vdelta)}^q - \psnorm[p]{\AA \vx}^q}{\psnorm[q]{\vx+ \vdelta}^q - \psnorm[q]{\vx}^q} \geq \frac{\langle \nabla g_n(\vx), \vdelta \rangle}{\langle \nabla g_d(\vx + \vdelta), \vdelta \rangle} = \psnorm[p]{\AA \vx}^{q-p} \frac{\langle \AA^T (\AA \vx)^{p-1}, \vdelta \rangle}{\langle \left( \vx+ \vdelta \right)^{q-1}, \vdelta \rangle}.
\end{equation*}
This bound will be used to design an iterate that satisfies the following lemma. We prove \autoref{lem:basic_invariant} in \autoref{apps:basic_invariant}.
\begin{lemma}
    \label{lem:basic_invariant}
    Given a matrix $\AA \in \rr^{m \times n}$, $\vx^{(0)} \in \rr^n$ such that $\psnorm[q]{\vx^{(0)}} = 1$, $0 < q,p$, $0 < \epsilon \leq 2/q$ and $V \geq 0$. Assume at each iteration we have
    \begin{equation*}
        \frac{\psnorm[p]{\AA \vx^{(t+1)}}^q - \psnorm[p]{\AA \vx^{(t)}}^q}{\psnorm[q]{\vx^{(t+1)}}^q - \psnorm[q]{\vx^{(t)}}^q} \geq \left(\left(1-\frac{\epsilon}{2}\right)V\right)^q,
    \end{equation*}
    then, once $\psnorm[q]{\vx^{(T)}}^q \geq \frac{4}{q \epsilon}$, we have 
    \begin{equation*}
        \frac{\psnorm[p]{\AA \vx^{(T)}}}{\psnorm[q]{\vx^{(T)}}} \geq (1-\epsilon)V.
    \end{equation*}
\end{lemma}
With scaling updates, i.e. $\vx^{(t+1)} = \vx^{(t)} + \vdelta^{(t)}$ where $\vdelta^{(t)}_i \in \lbrace 0, \alpha \vx_i^{(t)} \rbrace$ we obtain the following inequality.
\begin{equation*}
    \frac{\psnorm[p]{\AA \vx^{(t+1)}}^q - \psnorm[p]{\AA \vx^{(t)}}^q}{\psnorm[q]{\vx^{(t+1)}}^q - \psnorm[q]{\vx^{(t)}}^q} \geq \frac{1}{(1+\alpha)^{q-1}}\psnorm[p]{\AA \vx^{(t)}}^{q-p} \frac{\left\langle \AA^T \left( \AA \vx^{(t)} \right)^{p-1}, \vdelta^{(t)} \right\rangle}{\left\langle \left( \vx^{(t)} \right)^{q-1}, \vdelta^{(t)} \right\rangle}.
\end{equation*}

To simplify the notation, we define a vector of potentials $\vq(\vx) \in \rr^n$. Those potentials have nice properties, for instance, the direction of the gradient of $g$ can be extracted from the values of the potentials. Indeed, they are made such that $\nabla g(\vx) = \frac{\nabla g_d(\vx)}{g_d(\vx)} \circ \left(\vq(\vx)- g(\vx) \ones\right)$. That means when the potential is larger than $g(\vx)$, the gradient is positive. 
\begin{definition}
    \label{def:potentials}
    Given a positive vector $\vx \in \rr^n$, we define the potentials of $\vx$ as the vector $\vq(\vx) \in \rr^n$ such that
    \begin{equation*}
        \vq(\vx)_k = \psnorm[p]{\AA \vx}^{q-p} \frac{\left\langle \AA_{\sbullet,k}, \left(\AA \vx \right)^{p-1} \right\rangle}{\vx^{q-1}_k}.
    \end{equation*}
\end{definition}
The algorithm does not follow the traditional gradient descent framework, however each iteration still follows the direction of the gradient. We want our iterate to satisfy the following corollary of \autoref{lem:basic_invariant}.

\begin{corollary}
    \label{cor:first_invariant}
    Given a matrix $\AA \in \rr^{m \times n}$, $\vx^{(0)} \in \rr^n$ such that $\psnorm[q]{\vx^{(0)}} = 1$, $0 < q,p$, $0 < \epsilon < 2/q$ and $V \geq 0$, assume the following invariant is satisfied at each iteration
    \begin{equation}
        \label{eq:invariant_full}
        \vdelta^{(t)}_i = \alpha \vx_i^{(t)} \implies \frac{1}{\left(1+\alpha\right)^{q-1}} \vq(\vx)_i \geq \left(\left(1-\frac{\epsilon}{2}\right)V\right)^q,  
    \end{equation}
    then, once $\psnorm[q]{\vx^{(T)}}^q \geq \frac{4}{q\epsilon}$, we have 
    \begin{equation*}
        \frac{\psnorm[p]{\AA \vx^{(T)}}}{\psnorm[q]{\vx^{(T)}}} \geq (1-\epsilon)V.
    \end{equation*}
\end{corollary}

For the iterate to satisfy the invariant in \autoref{cor:first_invariant}, there must always be a potential larger than $V^q$. This is guaranteed by the following lemma from \cite{BV11}.

\begin{lemma}[Lemma 3.3 from \cite{BV11}]
    \label{lem:infeasibility_full}
    Let $\AA$ be a non-negative matrix and $q \geq p \geq 1$. Given $\vx > \zeros$, there is a coordinate $k$ such that $\vq(\vx)_k \geq \qpnorm[q][p]{\AA}^q$.
\end{lemma}

As long as $V \leq \qpnorm[q][p]{\AA}$, there always is a coordinate with potential above $V^q$. Setting $\alpha = \epsilon/2$, at time $t$ we scale coordinates with potential larger than $V^q$ so that our iterate satisfies Invariant~(\ref{eq:invariant_full}). At each iteration, a different coordinate may be scaled which means $\psnorm[q]{\vx^{(T)}}$ increases quite slowly. We only have the following lower bound $\psnorm[q]{\vx^{(T)}}^q \geq (1+\alpha)^{T/n}$, which requires $T = \OOt(n/q \epsilon)$ iterations are required to solve the approximate decision problem.

\subsection{A Width-Independent Algorithm}

In order to achieve an efficient algorithm, it is necessary to improve the number of itertaions by a factor of $n$. The challenge we had to deal before was that the $\lq$ norm of $\vx^{(t)}$ was not increasing rapidly enough. To overcome this issue, we will modify the algorithm to ensure there is a coordinate $k$ such that $\vx_k^{(t)}$ is large. Notice that \autoref{lem:infeasibility_full} guarantees that there is at least one coordinate with potential larger than $V^q$. The new invariant we would like to implement is that the set of coordinates with potential larger than $V^q$ does not take any new elements. This requires to also scale coordinates with a potential slightly smaller than $V^q$ which does not prevent satisfying the first Invariant~(\ref{eq:invariant_full}). In order to do so, we need to analyze how the potentials evolve between two iterations.

\begin{lemma}
    \label{lem:equivalent_potentials_full}
    Given a non-negative matrix $\AA \in \rr^{m \times n}$, $q \geq p \geq 1$ and a positive vector $\vx$, $\alpha > 0$. Let $\vdelta$ be such that $\vdelta_i = \alpha \vx_i$ or $\vdelta_i = 0$, we have
    \begin{enumerate}
        \item If $\vdelta_i \neq 0$, then $\vq(\vx + \vdelta)_i \leq \vq(\vx)_i$.
        \item If $\vdelta_i = 0$, then $\vq(\vx + \vdelta)_i \leq (1+\alpha)^{q-1} \vq(\vx)_i$.
    \end{enumerate}
\end{lemma}

\begin{proof}
    No matter the value of $\vdelta_i$, we have
    \begin{equation*}
        \vq(\vx + \vdelta)_i = \psnorm[p]{\AA \left(\vx + \vdelta\right)}^{q-p} \frac{\left\langle \AA_{\sbullet, i}, \left(\AA (\vx + \vdelta)\right)^{p-1} \right\rangle}{\left(\vx + \vdelta\right)_i^{q-1}}
        \leq (1+\alpha)^{q-1} \psnorm[p]{\AA \vx}^{q-p} \frac{\left\langle \AA_{\sbullet,i}, \left(\AA \vx\right)^{p-1} \right\rangle}{\left(\vx + \vdelta\right)_i^{q-1}}.
    \end{equation*}
    If $\vdelta_i = \alpha \vx_i$, then $\vq(\vx + \vdelta)_i \leq \vq(\vx)_i$, otherwise, $\vq(\vx + \vdelta)_i \leq (1+\alpha)^{q-1} \vq(\vx)_i$.
\end{proof}

As mentioned earlier, if we consider a multiplicative scaling $\vx^{(t+1)} = \vx^{(t)} + \vdelta^{(t)}$ where $\vdelta_i^{(t)} \neq 0$ if and only if the potential of coordinate $i$ is larger than $\theta$ for some $\theta$, then it creates a glass ceiling where no new coordinates can have potential noticeably larger than $\theta$. We quantify this upper bound in the following corollary.

\begin{corollary}
    \label{cor:second_invariant}
    Assume at each iteration we scale coordinates with potential larger than $\theta$, then the set of coordinates with potential larger than $(1+\alpha)^{q-1}\theta$ is decreasing.
\end{corollary}
\begin{proof}
    Let $C^{(t)}$ be the set of coordinates with potential larger than $(1+\alpha)^{q-1} \theta$ at time $t$. For a coordinate $i$ that is not in $C^{(t)}$, we have that if it is scaled at time $t$, then by \autoref{lem:equivalent_potentials_full}, $\vq(\vx^{(t+1)})_i \leq \vq(\vx^{(t)})_i$, hence $i$ is not part of $C^{(t+1)}$. Moreover, if $i$ is not scaled, then we know $\vq(\vx^{(t)})_i \leq \theta$, hence using \autoref{lem:equivalent_potentials_full}, $\vq(\vx^{(t+1)}) \leq (1+\alpha)^{q-1} \theta$. Thus, $i \notin C^{(t+1)}$. 
\end{proof}

We choose $\alpha$ and the threshold such that the glass ceiling is equal to $V^q$. That means that no new coordinates can have their potential larger than $V^q$. In the same time, we need to be careful about our choice of $\alpha$ and threshold so that Invariant~(\ref{eq:invariant_full}) is still satisfied.

Assume we hit coordinates with potential larger than $((1-\epsilon/4)V)^q$ with $\alpha = \epsilon/8$, we ensure that Invariant~(\ref{eq:invariant_full}) is satisfied since 
\begin{equation*}
    \frac{1}{\left(1+\frac{\epsilon}{8}\right)^{q-1}} \left(\left(1-\frac{\epsilon}{4}\right) V\right)^q \geq \left(\left(1-\frac{\epsilon}{2}\right)V\right)^q.
\end{equation*}

Moreover, using \autoref{cor:second_invariant}, we ensure that the set of coordinates with potential larger than $(1+\epsilon/8)^{q-1}(1-\epsilon/4)^q V^q < V^q$ is decreasing. We also know that at the end of the algorithm at least one coordinate has potential larger than $V^q$ using \autoref{lem:infeasibility_full}. Hence, those coordinates were scaled at each iteration and are equal to $(1+\alpha)^T \vx^{(0)}$. This gives the following theorem which we prove in \autoref{apps:algo_qpnorm_better}.

\begin{theorem}
    \label{thm:algo_qpnorm_better}
    Given a non-negative matrix $\AA \in \rr^{m \times n}$, a positive real $0 < \epsilon \leq 1/2q$, two reals $q \geq p \geq 1$, and a guess $V$ on $\qpnorm{\AA}$, \autoref{alg:qpnorm2} recovers $\vx$ such that 
    \begin{equation*}
        \frac{\psnorm[p]{\AA \vx}}{\psnorm[q]{\vx}} \geq (1-\epsilon)V
    \end{equation*}
    or certifies infeasibility in time 
    \begin{equation*}
        \OO\left(\frac{1}{q \epsilon} \log \left(\frac{n}{q \epsilon}\right) (\nnz(\AA) + m \log p)\right).
    \end{equation*}
\end{theorem}

\begin{algorithm}[H]
    \caption{Approximating $\lqp$-norm of non-negative matrices}
    \label{alg:qpnorm2}
    \begin{algorithmic}
        \Require $\AA \in \rr^{m \times n}_{\geq 0}$, $V$, $\epsilon$, and $q \geq p \geq 1$.
        \Ensure $\vx$ such that $\frac{\psnorm{\AA \vx}}{\psnorm[q]{\vx}} \geq (1-\epsilon) V$ or infeasibility certificate.
        \State $\vx \gets n^{-1/q} \ones$
        \State $\alpha \gets \epsilon /8$
        \While{$\frac{\psnorm{\AA \vx}}{\psnorm[q]{\vx}} < (1-\epsilon)V$}
        \State $\vdelta_i \gets \left\{
            \begin{array}{ll}
                \alpha \vx_i & \mbox{if } \vq(\vx)_i \geq ((1-\epsilon/4)V)^q \\
                0 &\mbox{otherwise} \\
            \end{array}
        \right.$
        \State $\vx \gets \vx + \vdelta$
        \If{for all $i, \vq(\vx)_i < V^q$}
            \State \textbf{return} ``infeasible''
        \EndIf
        \EndWhile
        \State \textbf{return} $\vx$
    \end{algorithmic}
\end{algorithm}

The algorithm can be parallelized since computing the potentials $\vq(\vx)$ requires computing two matrix-vector multiplications, hence we give as a corollary the running time obtained if we perform those operations in parallel.

\begin{corollary}
    \label{cor:parallel_running_time}
    Given a non-negative matrix $\AA \in \rr^{m \times n}$, a positive real $0 < \epsilon \leq 1/2q$, two reals $q \geq p \geq 1$, and a guess $V$ on $\qpnorm{\AA}$, it is possible to compute $\vx$ such that $\psnorm[p]{\AA \vx}/\psnorm[q]{\vx} \geq (1- \epsilon)\qpnorm[q][p]{\AA}$ in parallel time $\OOt(\frac{1}{q \epsilon})$ with total work $\OOt(\frac{\nnz(\AA)}{q \epsilon})$.
\end{corollary}

Using the Lewis weight sampling procedure from \cite{CohenP14,woodruff2022onlinelewisweightsampling}, we can obtain the following corollary which we prove in \autoref{apps:cor_algo_qpnorm_better}.

\begin{corollary}
    \label{cor:algo_qpnorm_better}
    Given a non-negative matrix $\AA \in \rr^{m \times n}$, a positive real $0 < \epsilon \leq 1/2q$, two reals $q \geq p \geq 1$, and a guess $V$ on $\qpnorm{\AA}$, by first preconditioning the input using \autoref{thm:lewis_sampling}, \autoref{alg:qpnorm2} returns $\vx$ such that 
    \begin{equation*}
        \frac{\psnorm[p]{\AA \vx}}{\psnorm[q]{\vx}} \geq (1-\epsilon)V
    \end{equation*}
    or certifies infeasibility in time 
    \begin{equation*}
        \OOt\left(\frac{r}{q \epsilon^3}n^{\max \lbrace \frac{p}{2},1 \rbrace} + \nnz(\AA) + n^\omega\right).
    \end{equation*}
    where $r$ is the maximum number of nonzero entries in a row of $\AA$.
\end{corollary}

\section{From Approximate Decision to Approximate Optimization}
\label{sec:approximation}
Our algorithm solves an approximate decision problem. When performing the binary search for the guess value $V$, a precision of $\epsilon$ is not required. We provide here a standard relaxation that transforms an approximate decision algorithm into an approximate optimization one.

First, notice that with $\kappa := \max_{i,j} \AA_{i,j}$, we have for any $q \geq p \geq 1$, $\kappa \leq \qpnorm{\AA} \leq \kappa n^{1-\frac{1}{q}} m^{\frac{1}{p}}$. Given a guess value $V$ and a precision $\epsilon$, \autoref{thm:algo_qpnorm_better} gives that either $\qpnorm{\AA}$ is in $[(1-\epsilon)V, +\infty]$ or $[0,V]$. We initialize our search interval as $[\kappa, n^{1-\frac{1}{q}} m^{\frac{1}{p}} \kappa]$.

Given a search interval $[L,U]$, we let $V = \sqrt{L U}$ and $\tilde{\epsilon} := \min \lbrace \frac{1}{2q}, \left(\frac{U}{L}\right)^{\frac{1}{6}}-1 \rbrace$. We run \autoref{alg:qpnorm2} with $\epsilon = \tilde{\epsilon}$ and $V$ as the guess value. If the algorithm returns a vector $\vx$ such that $\psnorm[p]{\AA \vx}/\psnorm[q]{\vx} \geq (1-\tilde{\epsilon})V$, we update $L = V(1-\tilde{\epsilon})$, otherwise we update $U = V$. We repeat this process until $U/L \leq 1/(1-\frac{\epsilon}{4})$. When this happens, we chose $V = L$ and use precision $\frac{\epsilon/4}{1+\epsilon/4}$. Since we maintain that $L$ is feasible and $U$ is infeasible, we obtain $\vx$ such that 
\begin{equation*}
    \frac{\psnorm[p]{\AA \vx}}{\psnorm[q]{\vx}} \geq \left(1-\frac{\epsilon/4}{1+\epsilon
    /4}\right)L \geq \frac{1-\frac{\epsilon}{4}}{1+\frac{\epsilon}{4}} U \geq (1-\epsilon)\qpnorm{\AA}.
\end{equation*}

We now analyse the cost of the search. While $\frac{U}{L} > \left(1+\frac{1}{2q} \right)^6$, we invoke the algorithm with precision $\frac{1}{2q}$ and $\log U/L$ is divided by a constant fraction. Hence, this process is repeated at most $\OO(\log \log (n^{1-\frac{1}{q}} m^{\frac{1}{p}}))$ times. Moreover, when $\frac{U}{L} \leq \left(1+\frac{1}{2q} \right)^6$, we use precision $\exp((\log U/L)/6) - 1 = \Theta(\log U/L)$. Since the running time of an iteration is proportional to $\frac{1}{\epsilon}$, the total running time of this part of the search is dominated by the last one with precision $\tilde{\epsilon} = \OO(\epsilon)$. Hence, the total running time is $\OO\left(\mathcal{T}(\frac{1}{2q}) \log \log (n^{1-\frac{1}{q}} m^{\frac{1}{p}}) + \mathcal{T}(\epsilon)\right)$ where $\mathcal{T}(\epsilon)$ is the running time of \autoref{alg:qpnorm2} with precision $\epsilon$.

\section{Computing an Optimal Oblivious Routing for a Fixed Monotonic Norm}
\label{sec:cpmrouting}

We explain in \autoref{apps:mcf_definitions} what multicommodity flows are and the assumptions we are making. Our goal is to find an optimal linear oblivious routing, i.e. a linear oblivious routing $\OBL^*$ such that 
\begin{equation*}
    \forall \OBL \in \Rlo, \CR(\OBL^*) \leq \CR(\OBL).
\end{equation*}

In the following, we assume the graph to be directed. Undirected graphs can be transformed in directed graphs, hence our results are also valid for undirected graphs (c.f. \autoref{apps:undirected_directed} for more details).

\begin{theorem}
    \label{thm:main_result_ccg}
    For any directed or undirected graph $G=(V,E)$ and any cost function that is a monotonic norm on the load vector, it is possible to compute a linear oblivious routing $\OBL$ such that 
    \begin{equation*}
        \CR(\OBL) \leq \min_{\OBL} \CR(\OBL) + \epsilon
    \end{equation*}
    in time $\OOt(n^6 m^3 \log \frac{n}{\epsilon})$ with high probability in $n$.
\end{theorem}

To construct the linear oblivious routing from \autoref{thm:main_result_ccg}, we will use the following lemma where $\mathcal{X}$ represents a relaxation of the set of linear routings and $\mathcal{Y}$ a relaxation of the set of demand vectors.
\begin{lemma}
    \label{lem:main_result_ccg}
    There exist two convex sets $\mathcal{X},\mathcal{Y} \subseteq \rr^{m \times n(n-1)}$ such that $\Rlo \subseteq \mathcal{X}$ and given $\widehat{\AA} \in \mathcal{X}$ and $\widehat{\QQ} \in \mathcal{Y}$ satisfying
    \begin{equation*}
        \max_{\QQ \in \mathcal{Y}} \langle \widehat{\AA}, \QQ \rangle_F - \min_{\AA \in \mathcal{X}} \langle \AA, \widehat{\QQ} \rangle_F \leq \epsilon,
    \end{equation*}
    we can construct in polynomial time a linear oblivious routing $\widetilde{\OBL}$ such that 
    \begin{equation*}
        \CR(\widetilde{\OBL}) \leq \min_{\OBL} \CR(\OBL) + \epsilon.
    \end{equation*}
\end{lemma}

To see how \autoref{lem:main_result_ccg} is useful, notice that using a convex-concave game solver we can find $\widehat{\AA}$ and $\widehat{\QQ}$ satisfying the condition of \autoref{lem:main_result_ccg}. The next theorem is a convex-concave game solver from \cite{cpm} that we will use to prove the running time of \autoref{thm:main_result_ccg}.

\begin{theorem}[Theorem C.9 from \cite{cpm}]
    \label{thm:concav_convex}
    Given convex sets $\mathcal{X} \subset B(0, R) \subset \rr^a$ and $\mathcal{Y} \subset B(0, R) \subset \rr^b$ such that both $\mathcal{X}$ and $\mathcal{Y}$ contain a ball of radius $r$. Let $H(\vx, \vy) : \mathcal{X} \times \mathcal{Y} \rightarrow \rr$ be an $L$-Lipschitz function that is convex in $\vx$ and concave in $\vy$. For any $0 < \epsilon \leq \frac{1}{2}$, we can find $(\hat{\vx}, \hat{\vy})$ such that 
    \begin{equation*}
        \max_{\vy \in \mathcal{Y}} H (\hat{\vx}, \vy) - \min_{\vx \in \mathcal{X}} H (\vx,\hat{\vy}) \leq \epsilon Lr
    \end{equation*}
    in time 
    \begin{equation*}
        \OO\left((a + b)^3 \log \left(\frac{a+b}{\epsilon} \frac{R}{r} \right) + (a + b) \mathcal{T} \log \left(\frac{a+b}{\epsilon} \frac{R}{r}\right)\right)
    \end{equation*}
    with high probability in $a + b$ where $\mathcal{T}$ is the cost of computing sub-gradient $\nabla f$.
\end{theorem}

\subsection{Relaxing the Problem}

\paragraph{Relaxing the set of linear routings.}

Notice that $\Rlo$ as defined in \autoref{def:routing} is convex, however it does not contain a ball of radius $r$ for $r >0$. Hence, we relax the problem to solving on $\mathcal{X}$ where 
\begin{equation*}
    \mathcal{X} := \left\lbrace \AA \in \rr^{m \times n(n-1)} :
    \begin{array}{l l c r}
        & \forall e \in E, (i,j) \in V \times V & \AA_{e,(ij)} &\geq 0 \\
        \wedge & \forall e \in E, (i,j) \in V \times V & \AA_{e,(ij)} &\leq 2 \\
        \wedge & \forall (i,j) \in V \times V & \sum_{e \in \OUT(i)} \AA_{e,(ij)} - \sum_{e \in \IN(i)} \AA_{e,(ij)} &\geq 1 \\
        \wedge & \forall (k,i,j) \in V \times V \times V & \sum_{e \in \OUT(k)} \AA_{e,(ij)} - \sum_{e \in \IN(k)} \AA_{e,(ij)} &\geq 0 \\
    \end{array}
    \right\rbrace.
\end{equation*}

We have that $\Rlo \subseteq \mathcal{X}$, and $\mathcal{X}$ contains a ball of radius of $r = m^{-\OO(1)}$ and is inside $B(0,2)$. The following claim is a result of $\norm{.}$ being monotone.

\begin{claim}
    \label{claim:relaxation_routing}
    Given a demand matrix $\DD$, we have 
    \begin{equation*}
        \min_{\OBL \in \Rlo} \norm{\OBL \DD} = \min_{\AA \in \mathcal{X}} \norm{\AA \DD}.
    \end{equation*}
\end{claim}

\paragraph{Relaxing the set of demand vectors.}

Let $\opt_{\leq 1}$ be the set of demand matrices $\DD$ such that $\min_{\OBL \in \Rlo} \norm{\OBL \DD} \leq 1$. Using \autoref{claim:relaxation_routing}, we have that $\opt_{\leq 1}$ can also be defined as the set of demand matrices $\DD$ such that $\min_{\AA \in \mathcal{X}} \norm{\AA \DD} \leq 1$. Since $\mathcal{X}$ only contains non-negative matrices, we have with $S_*$ the set of non-negative matrices with $\rr^{m \times K}$ coordinates and dual norm less than $1$ that $\norm{\AA \DD} = \max_{\GG \in S_*} \langle \AA, \GG \DD^T \rangle_F$. Let $\mathcal{Y}_1 := \lbrace \GG \DD^T : \GG \in S_*, \DD \in \opt_{\leq 1} \rbrace$, we want to find $\min_{\AA \in \mathcal{X}} \max_{\QQ \in \mathcal{Y}_1} \langle \AA,\QQ \rangle_F$. However, $\mathcal{Y}_1$ is not convex, we relax it to $\mathcal{Y} = \conv(\mathcal{Y}_1)$. The following claim is a consequence of $\QQ \mapsto \langle \AA, \QQ \rangle_F$ being linear.
\begin{claim}
    \label{claim:relaxation_demand}
    Given a matrix $\AA$ from $\mathcal{X}$, we have
    \begin{equation*}
        \max_{\QQ \in \mathcal{Y}_1} \langle \AA, \QQ \rangle_F = \max_{\QQ \in \mathcal{Y}} \langle \AA, \QQ \rangle_F.
    \end{equation*}
\end{claim}

\paragraph{Solving the relaxed problem.}

\begin{lemma}
    \label{lem:main_result_ccg_2}
    Assume $\norm{.}$ is a monotonic norm, let $\dnorm{.}$ be its dual norm. Define 
    \begin{equation*}
        \mathcal{X} := \left\lbrace \AA \in \rr^{m \times n(n-1)} :
        \begin{array}{l l c r}
            & \forall e \in E, (i,j) \in V \times V & \AA_{e,(ij)} &\geq 0 \\
            \wedge & \forall e \in E, (i,j) \in V \times V & \AA_{e,(ij)} &\leq 2 \\
            \wedge & \forall (i,j) \in V \times V & \sum_{e \in \OUT(i)} \AA_{e,(ij)} - \sum_{e \in \IN(i)} \AA_{e,(ij)} &\geq 1 \\
            \wedge & \forall (k,i,j) \in V \times V \times V & \sum_{e \in \OUT(k)} \AA_{e,(ij)} - \sum_{e \in \IN(k)} \AA_{e,(ij)} &\geq 0 \\
        \end{array}
        \right\rbrace,
    \end{equation*}
    and 
    \begin{equation*}
        \mathcal{Y} := \conv \lbrace \GG \DD^T : \DD \in \opt_{\leq 1}, \GG \in S_* \rbrace.
    \end{equation*}
    Given $\widehat{\AA} \in \mathcal{X}$ and $\widehat{\QQ} \in \mathcal{Y}$ such that
    \begin{equation*}
        \max_{\QQ \in \mathcal{Y}} \langle \widehat{\AA}, \QQ \rangle_F - \min_{\AA \in \mathcal{X}} \langle \AA, \widehat{\QQ} \rangle_F \leq \epsilon,
    \end{equation*}
    we can construct in polynomial time a linear oblivious routing $\widetilde{\OBL}$ such that
    \begin{equation*}
        \CR(\widetilde{\OBL}) \leq \min_{\OBL} \CR(\OBL) + \epsilon.
    \end{equation*}
\end{lemma}

To see how \autoref{lem:main_result_ccg_2} is helpful, notice that from $\widehat{\AA}$, we can construct $\widetilde{\OBL}$ that satisfies \autoref{def:routing} by reducing the coordinates of $\widehat{\AA}$. Since the cost function is monotone, for any $\QQ \in \mathcal{Y}$, $\langle \widetilde{\OBL}, \QQ \rangle \leq \langle \widehat{\AA}, \QQ \rangle$. Moreover, $\min_{\AA \in \mathcal{X}} \langle \AA, \widehat{\QQ} \rangle_F \leq \CR(\OBL^*)$ since the minimizer over $\mathcal{X}$ is in $\Rlo$. Hence, we have $\CR(\widetilde{\OBL}) \leq \CR(\OBL^*) + \epsilon$. Moreover, we get $\widehat{\AA},\widehat{\QQ}$ by using \autoref{thm:concav_convex}.

\section*{Acknowledgements}
This work was partially supported by the French Agence Nationale de la Recherche
(ANR), under grant ANR-21-CE48-0016 (project COMCOPT). We thank Alina Ene and Huy L\^{e} Nguy\~{\^{e}}n  for helpful conversations on approximating matrix norms.

\printbibliography

\appendix

\section{Multicommodity Flows}
\label{app:multicommodity_flows}

\subsection{Definitions}
\label{apps:mcf_definitions}

\paragraph{Multicommodity flow problem.}

We describe the general multicommodity flow framework introduced by Gupta et al. in \cite{Gupta06}. We are given $G=(V,E)$ an unweighted graph with $n$ vertices and $m$ edges and $\mathbf{R}$ a collection of routing requests $\mathbf{R} = \langle R_i = (s_i, t_i, d_i, k_i)\rangle$. Each routing request $R_i$ consists of a source-target pair $(s_i,t_i)$, a non-negative amount of traffic $d_i$ and a type of routing $k_i \in \llbracket 1 , K \rrbracket$. One can think of $k_i$ as a quality of service for instance. Each routing request is called a \emph{commodity} hence this problem is called the \emph{multicommodity} flow problem.

The cost of a multicommodity flow is a function of the load. For each edge $e \in E, \GG_{e,k}$ denotes the amount of flow of type $k$ that is sent along $e$. Each type of routing can induce different load on the same edge. Given a function $l : \rr^K \rightarrow \rr$, we can define the load $L(e)$ of the edge 
\begin{equation*}
    L(e) = l([\GG_{e,1}, \cdots, \GG_{e,K}]).
\end{equation*}
The cost is calculated by aggregating the edges via a function $\agg : \rr^m \rightarrow \rr$
\begin{equation*}
    \cost = \agg([L(e_1), \cdots, L(e_m)]).
\end{equation*}

This problem captures a lot of problems, for $l = \sum$ and $L = \max$ for instance, this captures the congestion of the framework.

Instead of using a collection of routing requests $\mathbf{R}$, we consider a demand matrix $\DD$ with $n(n-1)$ rows and $K$ columns. Each row corresponds to a directed pair of vertices and each column to a type of routing. The value of the matrix at row $i$ and column $k$ is the non-negative amount of traffic of type $k$ that needs to be routed from the source-target pair of the $i$-th line. We denote $\mathcal{D}$ the set of all demand matrices\footnote{Hence $\mathcal{D} = \lbrace \DD \in \rr^{n(n-1) \times K}_{\geq 0} : \DD \neq \zeros \rbrace$}. This notation allows to construct linear oblivious routings as matrices (hence linear).

\paragraph{Linear oblivious routings}

The multicommodity flow problem allows for a rather simple routing construction, for each source-target pair of vertices $(i,j)$ we route according to a precomputed unit non-negative flow $\vf_{(i,j)} \in \rr^m$ from $i$ to $j$. Those flows can be concatenated to form a non-negative matrix $\OBL \in \rr^{m \times n(n-1)}$. Given a demand matrix $\DD \in \mathcal{D}$, we call $\FF$ the flow induced by routing the demand $\DD$ with the oblivious routing $\OBL$. This flow matrix has $k$ columns and $m$ rows. Each one of its column is obtained by mulitplying $\OBL$ with a column of $\DD$: $\FF_{\sbullet,i} = \OBL \DD_{\sbullet,i}$. For the sake of simplicity, we write $\FF = \OBL (\DD)$. This flow matrix induces the load $L(e) = l(\OBL(\DD)_{e})$ on edge $e$. We call $\OBL$ a linear oblivious routing. Let $P_2^n$ be the set of oriented pairs of vertices and $P_3^n$ the set of oriented triplets of vertices.

\begin{definition}
    \label{def:routing}
    A matrix $\OBL$ is a linear oblivious routing if and only if it is a solution of the following system.
    \begin{align*}
        \forall e \in E, \forall (i,j) \in P_2^n : && \OBL_{e ,(i,j)} &\geq 0 \\
        \forall (i,j) \in P_2^n : && \sum_{e \in \OUT(i)} \OBL_{e , (i,j)} - \sum_{e \in \IN(i)} \OBL_{e, (i,j)} &= 1 \\
        \forall (k,i,j) \in P_3^n : && \sum_{e \in \OUT(k)} \OBL_{e, (i,j)} - \sum_{e \in \IN(k)} \OBL_{e, (i,j)} &= 0 \\
    \end{align*}
    We call $\Rlo$ the set of all linear oblivious routings on the graph $G$. A linear oblivious routing $\OBL$ on a demand $\DD$ creates a flow $\OBL(\DD)$ with cost 
    \begin{equation*}
        \cost(\OBL(\DD)) = \agg(L(l(\OBL(\DD)))).
    \end{equation*}
\end{definition}

In this report, we focus on the case where $\cost$ is a monotone norm $\norm{.}$ from $\rr^{m \times K}$ to $\rr$. 

\paragraph{Optimal flows and competitive ratio}

A multicommodity flow $\FF$ induces a cost $\agg(L(l(\FF)))$. The optimal cost of a demand matrix is $\opt(\DD) := \min_{R \in \Rlo} \cost(R \DD)$. The competitive ratio of the routing $\OBL$ is its worst performance against $\opt$. 

\begin{definition}
    The competitive ratio of $\OBL$ is defined as
    \begin{equation*}
        \CR(\OBL) = \max_{\DD \in \mathcal{D}} \frac{\cost(\OBL(\DD))}{\opt(\DD)}.
    \end{equation*}
\end{definition}

Since $\opt$ is linear, we may rescale the $\cost$ function to ensure that if $\cost(\DD) \leq 1$, then $\max_{i,j} \DD_{i,j} \leq 1$. This is possible since all norms are equivalent in finite dimension, moreover, this does not change the value of the competitive ratio. 

\begin{definition}
    The competitive ratio of $\OBL$ is 
    \begin{equation*}
        \CR(\OBL) = \max_{\DD \in \opt_{\leq 1}} \norm{\OBL(\DD)}
    \end{equation*}
    where $\opt_{\leq 1} := \lbrace \DD \in \mathcal{D} : \opt(\DD) \leq 1 \rbrace$.
\end{definition}

An optimal linear oblivious routing is a linear oblivious routing with the lowest competitive ratio among the competitive ratio of the others.

\subsection{From Undirected to Directed Graphs}
\label{apps:undirected_directed}

Consider an undirected graph $G = (V,E)$ with $n$ vertices and $m$ edges. We follow the same transformation as of \cite{Azar03}. For any undirected edge $(u,v)$, \cite{Azar03} consider the directed gadget $u,v,x,y$ with $5$ directed edges. The five directed edges are $e_1 = (u, x), e_2 = (v, x), e_3 = (y, u), e_4 = (y, v)$ and $e_5 = (x, y)$. We call $e_5$ the replacing edge of $(u,v)$. The new graph is a directed one, moreover, a linear oblivious routing on the new directed graph gives a linear oblivious routing on the undirected graph.

\section{Proofs from \texorpdfstring{\autoref{sec:preliminaries}}{ Preliminaries}}

\subsection{Proof of \texorpdfstring{\autoref{thm:concave_norm}}{ Theorem 3}}
\label{apps:concave_norm}

\begin{lemma}
    \label{lem:thm2_aux}
    For any non-negative vector $\vz \in \rr^n_{\geq 0}$, and $q \geq p \geq 1$, we have $g_{\vz} : \vu \mapsto \langle \vz, \vu^{1/q} \rangle^p$ is concave on $\rr^n_{\geq 0}$.
\end{lemma}

\begin{proof}[Proof of \autoref{lem:thm2_aux}]
    Given a non-negative vector $\vz$, let $f_{\vz} : \vx \mapsto \langle \vz, \vx^{1/q} \rangle^q$. Let $\vu,\vv \in \rr^n_{\geq 0}$ and $t \in [0,1]$. Define $\vx_i := t \vz_i^q \vu_i$ and $\vy_i := (1-t) \vz_i^q \vv_i$. We have 
    \begin{align*}
        f_{\vz}( t \vu + (1-t) \vv) &= \left(\sum_{i=1}^{n} \vz_i (t \vu_i + (1-t) \vv_i)^{1/q} \right)^q \\
        &= \left( \sum_{i=1}^{n} (t \vz_i^{q} \vu_i + (1-t) \vz_i^{q} \vv_i)^{1/q}\right)^q \\
        &= \left( \sum_{i=1}^{n} (\vx_i + \vy_i)^{1/q}\right)^q \\
        &\geq \left(\sum_{i=1}^{n} \vx_i^{1/q} \right)^q +  \left(\sum_{i=1}^{n} \vy_i^{1/q} \right)^q \\
        &= t \left(\sum_{i=1}^n \vz_i \vu_i^{1/q} \right)^q + (1-t) \left(\sum_{i=1}^{n} \vz_i \vv_i^{1/q} \right)^q \\
        &= t f_{\vz}(\vu) + (1-t) f_{\vz}(\vv).
    \end{align*}
    Where the inequality follows from Minkowski's inequality. Moreover, since $q \geq p$, $x \mapsto \sqrt[q/p]{x}$ is concave and non-decreasing, hence $g_{\vz} = \sqrt[q/p]{f_{\vz}}$ is concave.
\end{proof}

We are now ready to prove \autoref{thm:concave_norm}.

\begin{proof}
    Let $g : \vx \mapsto \psnorm[p]{\AA \vx^{1/q}}^{p}$. 
    Let $\vz_i$ be the $i$\textsuperscript{th} row of $\AA$. We have $g(\vx) = \sum_{i=1}^{m} g_{\vz_i}(\vx)$. Hence, by \autoref{lem:thm2_aux}, $g$ is concave. Moreover, since $x \mapsto \sqrt[p]{x}$ is concave and non-decreasing, we have that $\vx \mapsto \sqrt[p]{g(\vx)} = \psnorm[p]{\AA \vx^{1/q}}$ is concave.
\end{proof}

\section{Proofs from \texorpdfstring{\autoref{sec:lp_norm_algo}}{ Efficiently computing induced Lp-norms of non-negative matrices}}

\subsection{Proof of \texorpdfstring{\autoref{lem:basic_invariant}}{ Lemma 6}}
\label{apps:basic_invariant}

\begin{proof}
    We have
    \begin{align*}
        \frac{\psnorm[q]{\AA \vx^{(T)}}^q}{\psnorm[q]{\vx^{(T)}}^q} &= \frac{\psnorm[p]{ \AA \vx^{(0)}}^q + \sum_{t=0}^{T-1} \psnorm[p]{ \AA \vx^{(t+1)}}^q - \psnorm[p]{ \AA \vx^{(t)}}^q}{\psnorm[q]{\vx^{(T)}}^q} \\
        &\geq \frac{\psnorm[p]{\AA \vx^{(0)}}^q}{\psnorm[q]{\vx^{(T)}}^q} + \left(\left(1-\frac{\epsilon}{2}\right)V\right)^q \frac{\sum_{t=0}^{T-1}\psnorm[q]{\vx^{(t+1)}}^q - \psnorm[q]{\vx^{(t)}}^q}{\psnorm[q]{\vx^{(T)}}^q} \\
        &\geq \left(\left(1-\frac{\epsilon}{2}\right)V\right)^q \left(1-\frac{\psnorm[q]{\vx^{(0)}}^q}{\psnorm[q]{\vx^{(T)}}^q}\right)\\
        &= \left(\left(1-\frac{\epsilon}{2}\right)V\right)^q \left(1-\frac{1}{\psnorm[q]{\vx^{(T)}}^q}\right).
    \end{align*}
    Hence, once $1-\frac{1}{\psnorm[q]{\vx^{(T)}}^q} \geq (1-\epsilon/2)^q$, we have 
    \begin{equation*}
        \frac{\psnorm[p]{\AA \vx^{(T)}}}{\psnorm[q]{\vx^{(T)}}} \geq (1-\epsilon)V.
    \end{equation*}
    Since $\epsilon q \leq 2$, we have $(1-\epsilon/2)^q \leq 1-\epsilon q/4$, hence it is sufficient to have $\psnorm[q]{\vx^{(T)}}^q \geq 4/q \epsilon$.
\end{proof}

\subsection{Proof of \texorpdfstring{\autoref{thm:algo_qpnorm_better}}{ Theorem 11}}
\label{apps:algo_qpnorm_better}

\begin{proof}
    First, if the algorithm outputs ``infeasible'', we know by \autoref{lem:infeasibility_full} that it means $V > \qpnorm{\AA}$. Hence, assuming the algorithm does not return ``infeasible'', we will prove it outputs $\vx$ such that $\psnorm[p]{\AA \vx}/\psnorm[q]{\vx} \geq (1-\epsilon)V$.

    First, we prove that the invariant of \autoref{cor:first_invariant} is satisfied. At iteration $t$, if coordinate $i$ is scaled, then
    \begin{equation*}
        \frac{1}{\left(1+\alpha\right)^{q-1}} \vq(\vx^{(t)})_i \geq \frac{1}{\left(1+\frac{\epsilon}{8}\right)^{q-1}} \left(\left(1-\frac{\epsilon}{4}\right)V\right)^q \geq \left(\left(1-\frac{\epsilon}{2}\right)V\right)^q.
    \end{equation*}

    Moreover, using \autoref{cor:second_invariant}, we know that the set of coordinates with potential larger than $\left(\left(1-\frac{\epsilon}{4}\right)V\right)^q (1+\frac{\epsilon}{8})^{q-1} < V^q$ is decreasing. However, using \autoref{lem:infeasibility_full}, we know at least one coordinate has potential larger than $V^q$ at iteration $T$. Hence, this coordinate always had potential larger than $V^q$. Assume $k$ is such a coordinate, we have
    \begin{equation*}
        \psnorm[q]{\vx^{(T)}}^q \geq \left(\vx_k^{(T)}\right)^q = (1+\alpha)^{qT} n^{-1}. 
    \end{equation*}
    With $T \geq \frac{8}{q\epsilon} \log \frac{4n}{q\epsilon}$, we have $(1+\alpha)^{qT} n^{-1} \geq \frac{4}{q \epsilon}$ which implies 
    \begin{equation*}
        \frac{\psnorm[p]{\AA \vx^{(T)}}}{\psnorm[q]{\vx^{(T)}}} \geq (1-\epsilon)V
    \end{equation*}
    using \autoref{cor:first_invariant}.
\end{proof}

\subsection{Proof of \texorpdfstring{\autoref{cor:algo_qpnorm_better}}{ Corollary 13}}
\label{apps:cor_algo_qpnorm_better}

\begin{proof}
    Using \autoref{thm:lewis_sampling}, with $\tilde{\epsilon} = \epsilon/4$, we obtain a matrix $\AA'$ such that $(1- \tilde{\epsilon})\psnorm[p]{\AA' \vx} \leq \psnorm[p]{\AA \vx} \leq (1+\tilde{\epsilon})\psnorm[p]{\AA' \vx}$. Computing this matrix cost $\OOt(\nnz(\AA) + n^{\omega})$. Since $\AA$ is non-negative, then so is $\AA'$ as it corresponds to rescaled rows of $\AA$. Hence, we can use \autoref{thm:algo_qpnorm_better} on $\AA'$ with $\tilde{\epsilon}$ to obtain $\vx$ such that $\psnorm[p]{\AA' \vx} \geq (1-\tilde{\epsilon})\qpnorm[q][p]{\AA'}$. We can bound $\psnorm[p]{\AA \vx}$ by $\qpnorm[q][p]{\AA'}$.
    \begin{equation*}
        \psnorm[p]{\AA \vx} \geq (1-\tilde{\epsilon}) \psnorm[p]{\AA' \vx} \geq (1-\tilde{\epsilon})^2 \qpnorm[q][p]{\AA'}.
    \end{equation*}
    Moreover, we can bound $\qpnorm[q][p]{\AA'}$ by $\qpnorm[q][p]{\AA}$.
    \begin{equation*}
        \qpnorm[q][p]{\AA'} = \max_{\psnorm[q]{\vy} \leq 1} \psnorm[p]{\AA' \vy} \geq \max_{\psnorm[q]{\vy} \leq 1} \frac{1}{1+\tilde{\epsilon}} \psnorm[p]{\AA \vy} = \frac{1}{1+ \tilde{\epsilon}} \qpnorm[q][p]{\AA}.
    \end{equation*}
    Hence, $\psnorm[p]{\AA \vx} \geq \frac{(1-\tilde{\epsilon})^2}{1+\tilde{\epsilon}}\qpnorm[q][p]{\AA} \geq (1-\epsilon)\qpnorm[q][p]{\AA}$.

    Moreover, $\AA'$ has $\OOt\left(\frac{n^{\max \lbrace 1, p/2 \rbrace}}{\epsilon^2} \right)$ rescaled rows of $\AA$, hence $\nnz(\AA') \leq \OOt\left(\frac{r}{\epsilon^2} n^{\max \lbrace 1, p/2 \rbrace}\right)$ where $r$ is the maximum number of nonzero in a row of $\AA$.
\end{proof}
\end{document}